\documentclass{amsart}

\usepackage{amsmath}
\usepackage{amssymb}
\usepackage{amsfonts}
\usepackage{amsthm}

\usepackage[mathcal]{euscript}

\usepackage{epsf}
\usepackage{subfigure}
\usepackage{tikz}
\usetikzlibrary{calc}
\usetikzlibrary{arrows}
\usepackage{xcolor}
\definecolor{halfgray}{gray}{0.55}
\definecolor{webgreen}{rgb}{0,.5,0}
\definecolor{webbrown}{rgb}{.6,0,0}
\definecolor{Maroon}{cmyk}{0, 0.87, 0.68, 0.32}
\definecolor{RoyalBlue}{cmyk}{1, 0.50, 0, 0}
\definecolor{Black}{cmyk}{0, 0, 0, 0}

\usepackage[dots]{12many}
\usepackage[smaller]{acronym}
\usepackage{cancel}
\usepackage{color}
\usepackage{latexsym}
\usepackage{paralist}
\usepackage{ulem}
\usepackage{wasysym}

\usepackage[numbers,sort&compress]{natbib}

\usepackage{hyperref} 
\hypersetup{
colorlinks=true,
linktocpage=true,
pdfstartview=FitH,
breaklinks=true,
pdfpagemode=UseNone,
pageanchor=true,
pdfpagemode=UseOutlines,
plainpages=false,
bookmarksnumbered,
bookmarksopen=false,
bookmarksopenlevel=1,
hypertexnames=true,
pdfhighlight=/O,
urlcolor=webbrown, linkcolor=RoyalBlue, citecolor=webgreen, 
pdftitle={Dynamic Base Station Sharing},
pdfauthor={\textcopyright},
pdfsubject={},
pdfkeywords={},
pdfcreator={pdfLaTeX},
pdfproducer={LaTeX with hyperref}
}

\newcommand{\mb}{\mathbf}
\newcommand{\txs}{\textstyle}
\newcommand{\insum}{\sum\nolimits}

\newcommand{\pd}{\partial}

\newcommand{\simplex}{\Delta}

\newcommand{\set}{\mathcal{S}}
\newcommand{\play}{\mathcal{K}}
\newcommand{\act}{\mathcal{A}}

\newcommand{\strat}{\Delta}
\newcommand{\graph}{\mathcal{G}}
\newcommand{\edges}{\mathcal{E}}
\newcommand{\nodes}{\mathcal{V}}

\newcommand{\game}{\mathfrak{G}}

\newcommand{\eq}{\strat^{\!*}}

\newcommand{\cone}{T^{c}}

\newcommand{\R}{{\mathbb R}}

\DeclareMathOperator{\exclude}{\setminus}

\DeclareMathOperator{\tr}{tr}

\DeclareMathOperator{\Int}{Int}

\DeclareMathOperator{\ind}{ind}

\DeclareMathOperator{\eig}{eig}
\DeclareMathOperator{\rank}{rank}

\DeclareMathOperator{\supp}{supp}

\newcommand{\attn}[1]{{ #1}}

\newcommand{\negspace}{\!\!\!}
\newlength{\wideword}
\newlength{\leftover}
\theoremstyle{plain}
\newtheorem{theorem}{Theorem}
\newtheorem{corollary}[theorem]{Corollary}
\newtheorem*{corollary*}{Corollary}
\newtheorem{lemma}[theorem]{Lemma}
\newtheorem{proposition}[theorem]{Proposition}

\theoremstyle{definition}
\newtheorem{definition}[theorem]{Definition}
\newtheorem*{definition*}{Definition}

\theoremstyle{remark}
\newtheorem{remark}{Remark}
\newtheorem*{remark*}{Remark}
\newtheorem{example}{Example}

\begin{document}
\normalem


\title{Dynamic Power Allocation Games in\\Parallel Multiple Access Channels}

\date{\today}


\author[P.~ Mertikopoulos]{Panayotis Mertikopoulos}
\address{
Department of Economics\\
\'Ecole Polytechnique\\
91128 Palaiseau, France}
\email{panayotis.mertikopoulos@polytechnique.edu}

\author[E.~V.~Belmega]{Elena V.~Belmega}
\address{
Laboratoire des Signaux et Syst\`emes (LSS)\\
CNRS, Sup\'elec, Univ. Paris Sud 11\\
Palaiseau, France}
\email{belmega@lss.supelec.fr}
\thanks{The work of this author was partially supported by the French L'Or\'eal program ``For young women doctoral candidates in science'' 2009.}

\author[A.~L.~Moustakas]{Aris L.~Moustakas}
\address{
Department of Physics\\
University of Athens\\
15784, Athens, Greece}
\email{arislm@phys.uoa.gr}

\author[S.~Lasaulce]{Samson Lasaulce}
\address{
Laboratoire des Signaux et Syst\`emes (LSS)\\
CNRS, Sup\'elec, Univ. Paris Sud 11\\
Palaiseau, France
}
\email{lasaulce@lss.supelec.fr}

\maketitle

\begin{abstract}
We analyze the distributed power allocation problem in parallel multiple access channels (MAC) \attn{by studying an associated non-cooperative game which admits an exact potential function}. Even though games of this type have been the subject of considerable study in the literature \cite{SPB08i-sp,SPB08ii-sp,SPB08-it,SPB08-jsac}, we find that the sufficient conditions which ensure uniqueness of Nash equilibrium points typically do not hold in this context. Nonetheless, we show that the parallel MAC game admits a unique equilibrium almost surely, thus establishing an important class of counterexamples where these sufficient conditions are not necessary. Furthermore, if the network's users employ a distributed learning scheme based on the replicator dynamics, we show that they converge to equilibrium from almost any initial condition, even though users only have local information at their disposal.
\end{abstract}

\section{Introduction}
\label{sec:introduction}

As a result of the massive scale at which wireless networks are deployed and operate, non-cooperative game theory is rapidly becoming one of the main tools with which to describe and analyze \attn{distributed} resource allocation problems in this context. The reason for this is simple: whereas solution concepts and centralized optimization protocols which depend on global information are very hard to justify or implement (especially in real time or in the presence of a large number of users), game theory offers a way to look at the problem from a more distributed and localized point of view which is often of great applicational relevance.

A prime example of this can be seen in the huge corpus of literature surrounding power allocation games in \attn{static} Gaussian multi-user networks with the objective of reaching a Shannon-efficient state. The common characteristic of all these games is that the interference between multiple transmissions gives rise to non-trivial interactions between transmitters and imposes a bottleneck on the network performance: interference forces the power allocation policy of one user to depend on the power allocations of all other users. So, following \cite{lasaulce-spm-2009}, and given that the network users are left to optimally manage their resources on their own, the main questions that arise are
\begin{inparaenum}[\itshape a\upshape)]
\item whether there exist ``equilibrial'' allocations which are stable against unilateral deviations;
\item whether these (Nash) equilibria are unique; and
\item whether these equilibria can be reached by distributed (learning) algorithms which require only local information.
\end{inparaenum}

The two most important multi-user network models that have been studied from this perspective are the interference channel (IC) \cite{carleial-it-1978} and the multiple access channel (MAC) \cite{cover-book-2006}, two models which are inherently different from a communications point of view. On the one hand, the IC is composed of several non-cooperative trans\-mitter-receiver pairs and the information-theoretic capacity region is still an open issue for this channel model; in fact, even in the simple case of single-input, single-output (SISO) two-user Gaussian IC only the achievable rates are known \cite{han-it-1981,carleial-it-1978,sato-it-1981}. On the other hand, the MAC is composed of several transmitters and a single receiver which must decode the incoming messages, and its capacity region is relatively well-understood \cite{cover-book-1975,wyner-it-1974,cheng-it-1993}, \attn{something which remains an open problem for the IC}.

Perhaps the most general non-cooperative power allocation games \attn{studied in the context of static channels} are those presented by Scutari et al. in a series of seminal papers \cite{SPB08i-sp,SPB08ii-sp,SPB08-it,SPB08-jsac} focusing on the static Gaussian IC where receivers  employ the single-user decoding (SUD) scheme which treats incoming signals from other users as additive noise. There, the existence of a Nash equilibrium \attn{(in the ``pure'' sense of Rosen)} is a consequence of the convexity properties of the users' achievable rates and follows directly from Theorem 1 in \cite{Ro65}; in fact, under suitable (but stringent) conditions on the channel matrices, this equilibrium solution is unique.


Unfortunately, there are two issues with the approach of Scutari et al.: first, as the authors themselves admit, these sufficient conditions ``may not be easy to check'' \cite[p.~1925]{SPB09-sp} and, indeed, in most cases they are not (calculating the spectral radius of a matrix is very hard for large matrices). Secondly, these conditions are not necessary, so when they fail, the uniqueness issue is left wide open. In the specific case of two-user \attn{parallel} IC, some progress has been made in \cite{belmega-gamenets-2009}, where the authors completely characterize the set of Nash equilibria. Depending on the geometric properties of the best-response functions (which are identical to the water-filling operators of \cite{SPB08-it}), the power allocation game may have one, two, three or an infinite number of Nash equilibria.
Finally, in \cite{mochaourab-valuetools-2009}, assuming that the interference links in one of the bands are negligible, the game is shown to have strategic complementarities and the Nash set is studied using the super-modular property of the game.

A most interesting special case of these more general games consists of the parallel MAC power allocation games which are used to model uplink communication in multi-cellular wireless networks composed of several \attn{{\em nodes} (receivers, access points, base stations, etc.)} that operate in orthogonal frequency bands. From a mathematical point of view, the results of \cite{SPB08-it} obviously apply to the MAC as well, but, as we shall see, the sufficient conditions of \cite{SPB08-jsac} are never met in the \attn{parallel} MAC case, making them irrelevant to games of this type. To compensate for this, the authors of \cite{PBLD09} considered two different power allocation games in parallel multiple access channels, depending on the users' action sets:
\begin{inparaenum}[\itshape i\upshape)]
\item the users may distribute their available power among the \attn{wireless nodes}; or 
\item the users simply choose a node.
\end{inparaenum}
There, for the first game (which is more relevant for realistic power allocation scenaria), the Nash equilibrium is argued to be unique, but the proof provided in \cite{PBLD09} actually holds only under very restrictive conditions (otherwise, the authors' strict convexity arguments break down).

\smallskip

In this paper, we analyze non-co\-o\-pe\-ra\-tive power allocation games in parallel multiple access channels with the standard assumption of single user decoding (SUD)  at the receiver.\footnote{\attn{More efficient decoding techniques such as successive interference cancellation can also be considered \cite{belmega-springer-2010,belmega-twc-2009}, but optimality with respect to Shannon achievable rates will not concern us here; instead, the low level of signalling and decoder complexity of the SUD makes it more suitable for learning purposes. Furthermore, when using successive interference techniques, the exact potential property of the game is lost in general.}}
As in the more general MIMO MAC case, the parallel MAC game admits an exact potential (in the sense of \cite{MS96}) \attn{whose extrema correspond to the system's sum capacity, and which can also be interpreted as the system achievable sum-rate if users were employing successive interference cancellation (SIC)}. Since this potential function is convex, the game's Nash equilibria will correspond to the minima of the potential, so the game's Nash set is necessarily convex and compact. However, we find that the game's potential is, in general, {\em not} strictly convex (this was the mistake of \cite{PBLD09}), so one would expect that uniqueness of Nash equilibria fails along with the sufficient conditions of \cite{SPB08-jsac}. Rather surprisingly, we find that this is not the case: {\em even though the conditions of \cite{SPB08-jsac} do not hold, the Nash equilibrium of the game is unique (a.s.)}.


As far as convergence to equilibrium is concerned, one of the main results of \cite{yu-it-2004} is that if the transmitters know the local channel state and the overall interference-plus-noise covariance matrix, then the iterative sequential water-filling algorithm converges to the set of equilibrium points. On the other hand, asynchronous water-filling is harder to analyze because the sufficient conditions of \cite{SPB08-it} are typically not satisfied in the parallel MAC case. Finally, in a setting similar to our own (incorporating pricing but restricted to only one receiver), the authors of \cite{ABSA02} have considered update  algorithms which converge to equilibrium modulo certain conditions which do not always hold either.

Instead of taking a water-filling approach, we present a learning scheme based on the replicator dynamics of evolutionary game theory \cite{We95} which only requires the players to know their channel coefficients and their rates. Dynamics of this sort have been studied extensively in finite Nash games (\attn{that is, games with multilinear payoff functions over a strategy space which is a product of simplices} \textendash\ see e.g. \cite{FL98} for a survey) and in {\em continuous} population games \cite{We95,Sa01}, but, in the case of finite {\em nonlinear} games (such as the one we have here), their properties are not as well understood. The first step in that direction \attn{consists of identifying the correct modified version of the users' payoff functions which allows the replicator dynamics to behave well with respect to the solution concepts of the underlying game \textendash\ in more ``traditional'' finite player games, this purpose is served by the payoffs that correspond to the pure strategies of the game, but here we have no such structure.} Our main contribution is to then show that in parallel MAC power allocation games, the replicator dynamics converge to an equilibrium point unconditionally, {\em even in the zero-probability event where the game has multiple equilibria}.


\subsection*{Notational Conventions}
Throughout this paper, we will use bold uppercase letters to denote matrices and a dagger ``\dag'' to denote the Hermitian transpose of a complex matrix.

If $\set = \{s_{\alpha}\}_{\alpha=1}^{n}$ is a finite set, we will denote by $K\set$ the disjoint union (categorical coproduct) $K\set\equiv \coprod_{k=1}^{K}\set$ of $K$ copies of $\set$. Also, recall that the (real) vector space spanned by $\set$ is defined as the space $\R^{\set}\equiv\textrm{Hom}(\set,\R)$ of functions $x:\set\to\R$, equipped with the usual operations of addition and scalar multiplication of functions. The canonical basis $\{e_{\alpha}\}_{\alpha=1}^{n}$ of $\R^{\set}$ then consists of the indicator functions $e_{\alpha}:\set\to\R$ which take the value $e_{\alpha}(s_{\alpha})=1$ and vanish otherwise. Hence, under the natural identification $s_{\alpha}\mapsto e_{\alpha}$, we will use the index $\alpha$ to refer interchangeably to either $s_{\alpha}$ or $e_{\alpha}$, depending on the context. Similarly, we will also identify the set $\simplex(\set)$ of probability measures on $\set$ with the standard $\mbox{(n-1)}$-dimensional simplex of $\R^{\set}$: $\simplex(\set) \equiv \{x\in\R^{\set}: \insum_{\alpha} x_{\alpha} =1 \text{ and } x_{\alpha}\geq0\}$.

Finally, as far as players and their strategies are concerned, we will consistently employ Latin indices for players ($k,\ell,\ldots$), while reserving Greek ones for their (``pure'') strategies ($\alpha,\beta,\ldots$).

\section{The System Model}
\label{sec:model}

Following \cite{PBLD09}, the basic setup of our model is as follows: we have a set $\play=\ito{K}$ of finitely many wireless (single-antenna) transmitters \textendash\ the {\em players} of the game \textendash\ that wish to connect to a network of wireless nodes $\act=\ito{A}$ (for instance, a collection of base stations or access points). For simplicity, we are assuming that these nodes operate at distinct, non-interfering frequency bands, so that a user $k\in\play$ may split his transmitting power among the nodes $\alpha\in\act$ subject to the power constraint:
\begin{equation}
\label{eq:powerconstraint}
\insum_{\alpha} p_{k\alpha} \leq P_{k},
\end{equation}
where $p_{k\alpha}$ is the power with which user $k$ transmits towards node $\alpha$ and $P_{k}$ is the user's maximum transmitting power. As a result, the {\em power allocation} of the $k$-th user will be represented by the point $p_{k} = \sum_{\alpha} p_{k\alpha} e_{\alpha}\in\R^{\act}$, while, in obvious notation, the corresponding {\em power profile} which collectively reflects all of the users' power allocations will be represented by $p = (p_{1},\dotsc,p_{K})\in\R^{K\act}$.

Thus, under the standard assumption of single user decoding (SUD), the spectral efficiency of user $k$ in the power profile $p$ will be given by \cite{PBLD09,ABSA02}:
\begin{equation}
\label{eq:payoff}
u_{k}(p)
=\sum_{\alpha\in\act} u_{k\alpha}(p)
=\sum_{\alpha\in\act} b_{\alpha} \log\left(1 + \frac{g_{k\alpha}p_{k\alpha}}{\sigma_{\alpha}^{2} + \sum_{\ell\neq k} g_{\ell\alpha} p_{\ell\alpha}}\right),
\end{equation}
where:
\begin{enumerate}
\item $b_{\alpha}=B_{\alpha}/B>0$ is a normalized version of the bandwidth $B_{\alpha}$ of the node $\alpha\in\act$, rescaled to unity by the total bandwidth factor $B=\sum_{\alpha}B_{\alpha}$.
\item $g_{k\alpha}>0$ is the channel gain of user $k$ with respect to node $\alpha$, assumed here to be static for the duration of the transmission, known to user $k$, and drawn from a continuous (and nonatomic) probability distribution on the positive real numbers \textendash\ see also the relevant assumptions in \cite{SPB08-it,SPB09-sp}.
\item $\sigma_{\alpha}^{2}>0$ represents the noise level associated to node $\alpha$ (typically the variance of a Gaussian noise process).
\end{enumerate}

\begin{remark}
It should be noted here that when the wireless users are spatially distributed, the set $\act$ of wireless nodes need not be common to all users.
As it turns out, it is not too hard to extend our analysis and results to this more general case, but, to keep our presentation as clear as possible, we will only consider the case where every user can reach every node.
\end{remark}

\begin{remark}
We should also stress here that the channel gain coefficients $g_{k\alpha}$ are the only stochastic parameters in our model, \attn{and, in our static channel setting, they are given by $g_{k\alpha} = |h_{k\alpha}|^{2}$, where $h$ is a realization of the continuous random matrix which describes the channel \textendash\ see also \cite{PBLD09}}. So, unless explicitly mentioned otherwise, any probabilistic statement we make in this paper will refer to the probability law of the random variables $g_{k\alpha}$.
\end{remark}

\setcounter{remark}{0}

Now, as intuition would suggest (and as was shown rigorously in \cite{PBLD09}), when the users' utility is based solely on their spectral efficiency (\ref{eq:payoff}), it is clearly to the users' best interest to transmit at the highest possible total power, i.e. satisfying (\ref{eq:powerconstraint}) as an equality.\footnote{Of course, this need not be true if the cost of power consumption is too high \cite{ABSA02}, but we will not deal with this issue here.} As a result, we obtain the following components of a normal form game $\game$:
\begin{enumerate}
\item The set of {\em players} of $\game$ is $\play = \ito{K}$.
\item The {\em strategy space} of player $k$ is the (scaled) simplex $\strat_{k} \equiv \{p_{k}\in\R^{\act}: p_{k\alpha}\geq 0 \text{ and } \sum_{\alpha} p_{k\alpha}=P_{k}\}$; as is customary, we will denote the game's space of strategy profiles $p=(p_{1},\dotsc,p_{K})$ by $\strat\equiv\prod_{k}\strat_{k}$.
\item The players' {\em payoffs} (or {\em utilities}) are given by the spectral efficiencies $u_{k}:\strat\to\R$ of (\ref{eq:payoff}).
\end{enumerate}

Of course, the game $\game$ defined in this way is not \attn{finite} (in the original sense of \cite{Na51}) because
\begin{inparaenum}[\itshape a\upshape)]
\item the players are not mixing over a finite set of possible actions; and
\item even though the players' strategy spaces happen to be simplices, their payoffs are not multilinear over them.
\end{inparaenum}
On the other hand, since $\strat$ is a convex polytope and the utilities $u_{k}$ of the users are concave functions of their power allocations $p_{k}$, we immediately see that the game $\game$ is {\em concave} in the sense of Rosen \cite{Ro65}. Moreover, it was shown in \cite{PBLD09} that $\game$ is actually an {\em exact potential} game,\footnote{In the finite player sense of Monderer and Shapley \cite{MS96}, and not in the continuous sense of \cite{Sa01}.} i.e. that it admits a (global) potential function $\Phi:\strat\to\R$ such that:
\begin{equation}
\label{eq:potentialdef}
u_{k}(p_{-k};p_{k}') - u_{k}(p_{-k};p_{k}) = \Phi(p_{-k};p_{k}) - \Phi(p_{-k};p_{k}'),
\end{equation}
for all players $k\in\play$, and for all power allocations $p_{k},p_{k}'\in\strat_{k}$ of user $k$ and $p_{-k}\in\strat_{-k}\equiv\prod_{\ell\neq k}\strat_{\ell}$ of $k$'s opponents $\play_{-k}\equiv\play\exclude\{k\}$.\footnote{The change of signs in (\ref{eq:potentialdef}) from \cite{MS96} is deliberate. Our convention was chosen so as to conform with physics, where it is the {\em minima} of the potential function that are stable.}

In fact, the authors of \cite{PBLD09} provided the following explicit form for the potential function $\Phi$:
\begin{equation}
\label{eq:potential}
\Phi(p) = -\insum_{\alpha} b_{\alpha} \log\left(\sigma_{\alpha}^{2} + \insum_{k} g_{k\alpha} p_{k\alpha}\right).
\end{equation}
For posterity, note here that $\Phi$ is itself convex, but not necessarily {\em strictly} so:\footnote{This is precisely the subtle mistake that underlies the equilibrium uniqueness argumentation of \cite{PBLD09}.} indeed, any two power profiles $p,p'\in\strat$ such that $\insum_{k} g_{k\alpha}p_{k\alpha} = \insum_{k} g_{k\alpha} p_{k\alpha}'$ for all $\alpha\in\act$ will also have $\Phi(p) = \Phi(p')$. This simple observation will be of crucial importance in determining the Nash set of the game, so we will pause here to introduce the concept of {\em degeneracy}.

\smallskip

\begin{figure*}
\label{fig:degeneracy}
\centering
\subfigure[$\ind(\game)=0$: generic level sets]{%
\label{subfig:contirr}
\includegraphics[width=0.45\columnwidth]{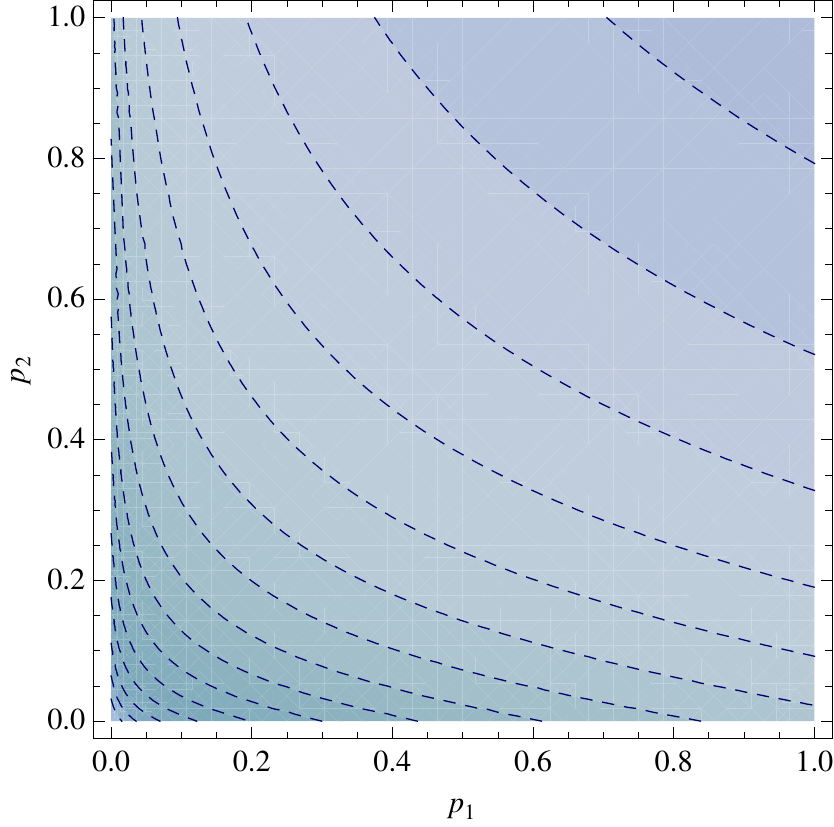}}
\hfill
\subfigure[$\ind(\game)>0$: degeneration into affine sets]{
\label{subfig:contred}%
\includegraphics[width=0.45\columnwidth]{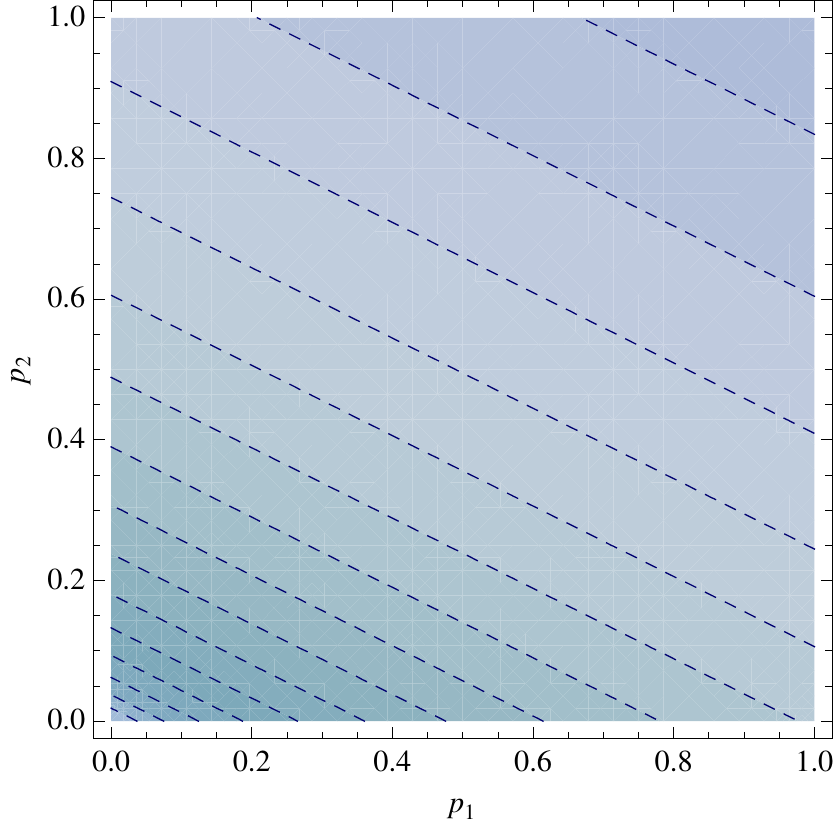}}
\caption{The level sets (dashed blue lines) of the potential function $\Phi$ in a $2\times2$ game with and without degenerate directions (Figs.~\ref{subfig:contred} and \ref{subfig:contirr} respectively). Degeneracy nullifies degrees of freedom and introduces redundant directions in the system.}
\end{figure*}

To that end, let $T_{p}\strat$ denote the tangent space of $\strat$ at $p$. Since $\strat$ is an affine polytope embedded in $\R^{K\act}$, it is easy to see that for every interior point $p\in\Int(\strat)$, $T_{p}\strat$ will be isomorphic to the subspace which is ``parallel'' to the polytope $\strat$:
\begin{equation}
\txs
T_{p}\strat \cong Z = \left\{z\in\R^{K\act}:\sum_{\alpha} z_{k\alpha} = 0\text{ for all $k\in\play$}\right\}.
\end{equation}
However, as we just noted, some of these $KA-K$ directions will be {\em degenerate} (or {\em redundant}), in the sense that the potential $\Phi$ remains constant as we move along them. Specifically, the set of (almost surely independent) constraints
\begin{equation}
\label{eq:redconstraints0}
\txs
\insum_{k}g_{k\alpha} z_{k\alpha} = 0,\,
\alpha\in\act,
\end{equation}
cuts itself a $(KA-A)$-dimensional subspace $W$ of $R^{K\act}$ whose intersection with $Z$ will correspond to the total of $K+A$ constraints:
\begin{subequations}
\label{eq:constraints}
\begin{align}
\label{eq:tanconstraints}
a)\quad	&\insum_{\alpha} z_{k\alpha} = 0,	&k\in\play;\\
\label{eq:redconstraints}
b)\quad	&\insum_{k} g_{k\alpha} z_{k\alpha} = 0,	&\alpha\in\act.
\end{align}
\end{subequations}


Of course, the $K$ tangent space constraints (\ref{eq:tanconstraints}) are set in stone while the $A$ degeneracy constraints (\ref{eq:redconstraints}) depend on the realization of the channel gains $g_{k\alpha}$.\footnote{These conditions are remarkably similar to the MIMO rank condition $\rank (\mb{H}^{\dag}\mb{H})=\sum_{k=1}^K n_{t,k} \leq n_r + K$ (where $\mb{H} = [\mb{H}_1, \hdots, \mb{H}_K]$ is the system's channel matrix) which ensures that there exists a unique Nash equilibrium \cite{BelmegaThesis}.} It is thus possible (though improbable) that some of the constraints (\ref{eq:constraints}) are linearly dependent. To keep track of all this, we have:
\begin{definition}
The subspace $W\leq\R^{K\act}$ defined by the constraints (\ref{eq:redconstraints}) will be called the space of {\em degenerate} (or {\em redundant}) directions of the game $\game$. Moreover, we define the {\em degeneracy} (or {\em redundancy}) {\em index} of $\game$ to be:
\begin{equation}
\label{eq:degeneracy}
\ind(\game) \equiv \dim(W\cap Z),
\end{equation}
where $Z$ is the tangent space determined by the admissibility constraints (\ref{eq:tanconstraints}).
\end{definition}

\begin{example}
As we just saw, $\ind(\game) = KA - K - A$ (a.s.), so there is no degeneracy in games with $K=2$ users and $A=2$ nodes. However, if the channel gains of the two users happen to be linearly dependent (a zero-probability event to be sure, but one which could be approximated reasonably well by strongly collocated users), then we can have degenerate directions even in a $2\times2$ game (see Fig.~\ref{fig:degeneracy}). In that case, the potential function $\Phi$ ceases to be strictly convex, so there is no a priori reason that the potential's minimum set will be a singleton.
\end{example}

\section{Equilibrium Analysis}

In this section, our main goal will be to describe the Nash set of the game and, more specifically, to show that it consists (almost surely) of a unique equilibrium point which is located at a face of the strategy space $\strat$.


\attn{
This problem has attracted considerable interest in the literature, where the papers by Scutari et al. \cite{SPB08i-sp,SPB08ii-sp,SPB08-it,SPB08-jsac} provide a set of sufficient conditions for uniqueness in more general interference channel scenaria, and, more recently, in \cite{PBLD09}, where the authors focus on the parallel MAC problem. Unfortunately, these approaches are problematic (for different reasons): on the one hand, we will see that the sufficient conditions of Scutari et al. \cite{SPB08i-sp,SPB08ii-sp,SPB08-it,SPB08-jsac} do not hold in our setting; on the other hand, the uniqueness proof of \cite{PBLD09} is only valid under the extremely restricting condition that the game is {\em non-degenerate}, i.e. that $KA\leq K+A$ (otherwise, the potential $\Phi$ is not strictly convex).

Indeed, especially this last condition holds for very few scenarios, only one of which is (barely) non-trivial:
\begin{inparaenum}[\itshape a\upshape)]
\item when we have $K=1$ user and an arbitrary number $A\geq1$ of nodes (in which case the problem reduces to an optimization one solved by water-filling \cite{boyd-book-2004});
\item when an arbitrary number of $K\geq1$ users transmits with the maximum possible power to a single node shared by all ($A=1$; this was also the scenario studied by \cite{ABSA02} who also introduced a linear pricing scheme to compensate for power costs); and
\item in the $2x2$ case which is easy to solve directly.
\end{inparaenum}
In spite of the above, our main result in this section is that the sufficient conditions of Scutari are actually far from necessary: {\em for (almost) any realization of the channel gain parameters $g_{k\alpha}$, there exists a unique Nash equilibrium.}
}

\subsection{Nash Equilibrium Conditions}

Since we have a finite number of players, the notion of Nash equilibrium takes the form of stability in the face of unilateral deviations. More specifically:
\begin{definition}
We will say that the power profile $q\in\strat$ is at {\em Nash equilibrium} in the game $\game$ when
\begin{equation}
\label{eq:Nash}
u_{k}(q) \geq u_{k}(q_{-k};q_{k}'),
\end{equation}
for all $k\in\play$, and for every deviation $q_{k}'\in\strat_{k}$ of player $k$.

In particular, if $q$ satisfies the strict version of the inequalities (\ref{eq:Nash}), then it will be called a {\em strict equilibrium} of $\game$.
\end{definition}

As is standard in convex potential games, to calculate the Nash set $\eq\equiv\eq(\game)$ of the game, we only need to look at the (necessarily convex) minimum set of the potential function $\Phi$. To that end, the first order constrained Karush-Kuhn-Tucker minimization conditions \citep{PBLD09} show that a power profile $q\in\strat$ will be at Nash equilibrium if and only if:
\begin{subequations}
\label{eq:KKT}
\begin{flalign}
a)\quad	&\lambda_{k} - \frac{b_{\alpha}g_{k\alpha}}{\sigma_{\alpha}^{2} + \insum_{\ell} g_{\ell\alpha} q_{\ell\alpha}}\geq 0\\
b)\quad	&q_{k\alpha} \left(
\lambda_{k} - \frac{b_{\alpha}g_{k\alpha}}{\sigma_{\alpha}^{2} + \insum_{\ell} g_{\ell\alpha} q_{\ell\alpha}}\right)
= 0,
\end{flalign}
\end{subequations}
for all players $k\in\play$ and all nodes $\alpha\in\act$ (and with the obvious constraints imposed by the condition $q\in\strat$).

An obvious observation that can be gleaned from the above is that if $q$ is a Nash equilibrium, then either
\begin{inparaenum}[\itshape a\upshape)]
\item the support $\supp(q_{k})\equiv\{\alpha\in\act: q_{k\alpha}>0\}$ of a user's power allocation is a singleton (i.e. the user only transmits to a single node); or
\item we will have the ``waterfilling'' condition:
\begin{equation}
\label{eq:waterfill}
\frac{g_{k\alpha}}{g_{k\beta}} = \frac{r_{\alpha}}{r_{\beta}}\quad
\text{for all $\alpha,\beta\in\supp(q_{k})$},
\end{equation}
where $r_{\alpha}$ is the user-independent quantity given by:
\begin{equation}
\label{eq:rdef}
r_{\alpha}^{-1}
\equiv \frac{b_{\alpha}}{\sigma_{\alpha}^{2} + \insum_{\ell} g_{\ell\alpha} q_{\ell\alpha}}.
\end{equation}
In other words, if user $k$ connects to more than one node and is at equilibrium, then he must be ``waterfilling'' the quantity $g_{k\alpha}/r_{\alpha}$ among the nodes that he employs.
\end{inparaenum}

\smallskip

A promising way to determine whether our game admits a {\em unique} equilibrium is to take advantage of the plethora of sufficient conditions that have been established in the literature for this purpose. In our setting, the condition which is easiest to check was the one proposed by Scutari et al. in \cite[Equation (21)]{SPB08-jsac}, and which takes the form:
\begin{equation}
\label{eq:maxcondition}
\tag{Cmax}
\rho(\mb S_{\text{max}}) < 1,
\end{equation}
where $\rho(\mb S_{\text{max}})=\max\{|\lambda|:\lambda\in\eig(\mb S_{\text{max}})\}$ is the {\em spectral radius} (i.e. the eigenvalue with the largest modulus) of the $K\times K$ matrix $\mb S_{\text{max}} = \{S^{\text{max}}_{k\ell}\}$ defined as:
\begin{equation}
\label{eq:Smax}
S^{\text{max}}_{k\ell} =
\begin{cases}
0,	&\quad k=\ell,\\
\max_{\alpha}\left\{g_{\ell\alpha}\big/g_{k\alpha}\right\},		&\quad k\neq\ell.
\end{cases}
\end{equation}

However, since $\max_{\alpha}\{g_{\ell\alpha}/g_{k\alpha}\} = \left(\min_{\alpha}\{g_{k\alpha}/g_{\ell\alpha}\}\right)^{-1}\geq \left(S^{\text{max}}_{\ell k}\right)^{-1}$, we immediately see that the entries of $\mb S_{\text{max}}$ satisfy the inequality $S_{k\ell}^{\text{max}} S_{\ell k}^{\text{max}}\geq1$ for any distinct pair of users $k,\ell\in\play$. Hence, $\tr(\mb S_{\text{max}}^{2})$ will be bounded from below by:
\begin{equation}
\label{eq:tracebound}
\tr(\mb S_{\text{max}}^{2}) = \insum_{k,\ell} S^{\text{max}}_{k\ell} S^{\text{max}}_{\ell k} \geq K(K-1),
\end{equation}
and, by the spectral radius bounds of \cite{Horne97}, we will have:
\begin{equation}
\label{eq:spectralbound}
\rho(\mb S_{\text{max}}) \geq \frac{|\tr(\mb S_{\text{max}})|}{S} + \sqrt{\frac{\tr(\mb S_{\text{max}}^{2}) - \tr(\mb S_{\text{max}})^{2}/S}{S(S-1)}},
\end{equation}
where $S=\rank(\mb S_{\text{max}})$.\footnote{Strictly speaking, (\ref{eq:spectralbound}) holds if $S\geq 2$, but we can trivially disregard the case $S<2$, because $S=K$ almost surely and the one-user case holds little interest.} However, since $\tr(\mb S_{\text{max}})=0$ by definition, (\ref{eq:spectralbound}) gives $\rho(\mb S_{\text{max}})\geq 1$, so the sufficient condition (\ref{eq:maxcondition}) fails.


Other sufficient conditions were put forth in \cite{SPB08i-sp,SPB08ii-sp,SPB08-it} and \cite{LP06}, based on the matrices $\mb S(\alpha)=\{S_{k\ell}(\alpha)\}$ defined as:
\begin{equation}
S_{k\ell}(\alpha) =
\begin{cases}
0,						&k=\ell,\\
g_{\ell\alpha}/g_{k\alpha},	&k\neq\ell.
\end{cases}
\end{equation}
In particular, it was shown in \cite{SPB08i-sp} that if:
\begin{equation}
\label{eq:spectralcondition}
\tag{C1}
\rho(\mb S(\alpha)) <1\quad
\text{for all $\alpha\in\act$},
\end{equation}
then the game $\game$ admits a unique Nash equilibrium; in a similar vein, the authors of \cite{LP06} proposed the condition:
\begin{equation}
\label{eq:positivecondition}
\tag{C2}
\mb I + \mb S(\alpha) \succcurlyeq 0\quad
\text{for all $\alpha\in\act$},
\end{equation}
where ``$\succcurlyeq 0$'' signifies positive-definiteness.

Of these two conditions (\ref{eq:spectralcondition}) is {\em stronger} than (\ref{eq:maxcondition}) in the sense that (\ref{eq:maxcondition}) is sufficient for (\ref{eq:spectralcondition}). However, the same analysis as before shows that in the case of the $\mb S(\alpha)$ matrices, (\ref{eq:tracebound}) holds as an equality, so we still get $\rho(\mb S(\alpha))\geq 1$ for all $\alpha\in\act$, causing (\ref{eq:spectralcondition}) to fail. Similarly, even though the positive-definiteness condition (\ref{eq:positivecondition}) is independent of (\ref{eq:spectralcondition}) and (\ref{eq:maxcondition}), the definition of $\mb S(\alpha)$ yields $S_{k\ell}(\alpha) + S_{\ell k}(\alpha)\geq 2$ for all $k\neq\ell$. Consequently, the element with the largest modulus of the symmetrized matrix $\mb I + \frac{1}{2}(\mb S(\alpha) + \mb S^{\dag}(\alpha))$ does not lie on the main diagonal, so the matrix $\mb I + \mb S(\alpha)$ cannot be positive-definite either.

\begin{remark*}
Strictly speaking, condition (\ref{eq:spectralcondition}) was phrased in \cite{SPB08i-sp} in terms of  a slightly different version of the matrix $\mb S(\alpha)$ where $S_{k\ell}(\alpha)=0$ whenever the channel of $\alpha$ is ``too bad'' for either $k$ or $\ell$ (in a sense made precise in \cite{SPB08i-sp}). In this more general setup, if $\alpha$ is ``bad'' for user $k$, then the $k$-th row and $k$-th column of $\mb S(\alpha)$ vanish, so the bound (\ref{eq:tracebound}) is decreased to $(K-r)(K-r-1)$, where $r$ is the number of zero rows and columns that were introduced in $\mb S(\alpha)$. However, this also reduces the rank of $\mb S(\alpha)$ accordingly, so, assuming that $\rank(\mb S(\alpha))\geq2$, the bound (\ref{eq:spectralbound}) still gives $\rho(\mb S(\alpha))\geq1$.

Of course, this still leaves open a small window where the condition (\ref{eq:spectralcondition}) might be salvaged \textendash\ namely the rare occurence where the $\mb S(\alpha)$ matrices all have rank $1$ or less. However, instead of focusing on this very special case, we note that even the extensive numerical simulations of \cite{SPB08i-sp} show that the sufficient condition (\ref{eq:spectralcondition}) almost never holds in the parallel MAC setting. Indeed, if we follow \cite{SPB08i-sp} and assume for simplicity that the transmitter-receiver distances are all equal ($d_{qr} = d_{rq}$ in their notation), then the ``normalized interlink distance'' becomes equal to $1$ and Figure 1 of \cite{SPB08i-sp} reveals that (\ref{eq:spectralcondition}) fails almost surely.
\end{remark*}

We thus see that, despite their theoretical value, the sufficient conditions that have been established in the literature are quite problematic in the parallel MAC setting because they are typically never met (except possibly in some very special cases). Therefore, in order to address the uniqueness issue in complete generality, we will need to develop a different set of tools.

\subsection{Representing Power Profiles as Graphs}
As we shall see, the ``waterfilling'' conditions (\ref{eq:waterfill}) impose some pretty severe constraints on the structure of the equilibrium set $\eq$, because whenever a user waterfills between nodes, his channel gains must ``split'', i.e. be of the form $g_{k\alpha} = \lambda_{k} r_{\alpha}$. This is actually best understood pictorially, by representing a power profile $p\in\strat$ as a graph:

\begin{definition}
\label{def:graph}
We will say that the (multi)graph $\graph\equiv(\nodes,\edges)$ {\em represents} the power profile $p=(p_{1},\dotsc,p_{K})\in\strat$ if:
\begin{enumerate}
\item $\nodes=\act$: the nodes of $\graph$ coincide with the network's;
\item for each $k\in\play$, there is a node $\alpha\in\act$ (called the {\em hub} of user $k$ in $\graph$) to which the user assigns positive power $p_{k\alpha}>0$, and which is joined by an edge of $\graph$ to every {\em other} node $\beta\in\supp(p_{k})\exclude\{\alpha\}$.
\end{enumerate}
\end{definition}

In simpler words, to represent a power profile $p\in\strat$ as a graph, one merely has to take the set of wireless nodes as the set of the graph's nodes, and then, for every user $k\in\play$, to pick a node which the user employs and connect it with an edge to every other node to which the user assigns positive power. Of course, depending on the choice of ``hub'' for each user $k\in\play$, one might end up with non-isomorphic graphs representing the same power profile $p$. However, this lack of uniqueness will not be important to us, so we will occasionally abuse Definition \ref{def:graph} by using $\graph(p)$ to collectively denote {\em any} graph which represents the profile $p\in\strat$.

\smallskip

In light of the above, we now state a few key lemmas and corollaries that will be crucial in our efforts to understand the structure of the equilibrial set $\eq$. The first one is an elegant structural property of equilibrial graphs:

\begin{lemma}
\label{lem:graph}
Let $\graph\equiv\graph(p)$ represent a power profile $p\in\eq$ which is at Nash equilibrium. Then $\graph$ is almost surely a forest \textendash\ that is, $\graph$ contains no cycles.
\end{lemma}

\begin{proof}
The intuitive idea behind this lemma is that if there is a cycle, then we can get a chain of fractions $g_{k_{1},\alpha_{1}}/g_{k_{0},\alpha_{0}}$, $g_{k_{2},\alpha_{2}}/g_{k_{2},\alpha_{1}}$, $\dotsc$, which will have a product equal to $1$ because of the waterfilling condition (\ref{eq:waterfill}). However, this represents a condition on the $g$'s which occurs with zero probability, thus providing a contradiction.

To make this idea precise, assume that $\graph$ contains a cycle $\Gamma$ denoted as a sequence of edges $\Gamma = (e_{1},\dotsc,e_{n})$.\footnote{Note that keeping track only of the nodes is not enough because two distinct edges might link the same pair of nodes.} Since an edge can only be owned by a single player, this cycle gives rise to a sequence of players which we also denote by $\Gamma = (k_{1},\dotsc,k_{n})$.

So, if $(\alpha_{0},\alpha_{1},\dotsc,\alpha_{n})$ is the corresponding sequence of nodes that $\Gamma$ passes through (obviously, $\alpha_{0} = \alpha_{n}$), then (\ref{eq:waterfill}) gives:
\begin{equation}
\frac{g_{k_{j},\alpha_{j}}}{g_{k_{j},\alpha_{j-1}}} = \frac{r_{\alpha_{j}}}{r_{\alpha_{j-1}}},
\text{ for all $j=\ito{n}$.}
\end{equation}
Therefore, multiplying these $n$ equations together, we get:
\begin{equation}
\label{eq:chain}
\frac{g_{k_{1},\alpha_{1}}}{g_{k_{0},\alpha_{0}}} \dotsm \frac{g_{k_{n},\alpha_{n}}}{g_{k_{n-1},\alpha_{n-1}}}
=\frac{r_{\alpha_{1}}}{r_{\alpha_{0}}} \dotsm \frac{r_{\alpha_{n}}}{r_{\alpha_{n-1}}}
=1.
\end{equation}
Since there are no cancellations in this last equation (recall that the nodes $a_{j}$, $j=1,\dotsc,n-1$ of $\Gamma$ are all distinct), it will describe a measure zero submanifold of the space from which the channel coefficients $g$ are drawn. As a result (\ref{eq:chain}) only holds with probability zero and, hence, the assumption that $\graph$ contains a cycle is almost surely false.
\end{proof}

\begin{figure}
\centering
\begin{tikzpicture}
[nodestyle/.style={circle,draw=black,fill=gray!5, inner sep=2pt},>=stealth]

\coordinate (A) at (-2,0);
\coordinate (B) at (0,1.5);
\coordinate (C) at (2,0);
\coordinate (D) at (0,-1.5);

\node (A) at (A) [nodestyle] {$\alpha$};
\node (B) at (B) [nodestyle] {$\beta$};
\node (C) at (C) [nodestyle] {$\gamma$};
\node (D) at (D) [nodestyle] {$\delta$};

\draw [red, densely dashed] (A) to [bend right = 15] (B);
\draw [red, densely dashed] (A) to (D);
\draw [blue] (B) to [bend right=15] (A);
\draw [blue] (B) to (C);
\draw [webgreen,densely dotted] (D) to (C);
\draw [webgreen,densely dotted] (D) to (B);

%

\end{tikzpicture}
\caption{A graph representing a power profile in a game with 3 users (red, blue and green) and 4 nodes ($\alpha,\beta,\gamma$ and $\delta$). In the profile represented above, the red player uses $\alpha,\beta$ and $\delta$, blue uses $\alpha,\beta$ and $\gamma$, and green employs $\beta,\gamma$ and $\delta$.}
\end{figure}
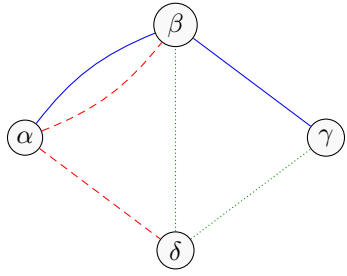

As an easy corollary of the above lemma, we also get:
\begin{corollary}
\label{cor:maxwaterfill}
If $p\in\eq$ is an equilibrial power profile, then there are (a.s.) at most $A-1$ instances of waterfilling (i.e. two nodes employed by the same player).
\end{corollary}

\begin{proof}
Simply note that a forest on $A$ nodes can have at most $A-1$ edges; our result then follows by recalling that an edge in this context simply represents an instance of waterfilling.
\end{proof}

From a geometrical point of view, this shows that Nash equilibria can only live on the faces of the strategy space $\strat$:

\begin{corollary}
\label{cor:face}
Let $p\in\eq$ be an equilibrial power profile. Then $p$ lies in the interior of an at most $(A-1)$-dimensional face of $\strat$ (a.s.).
\end{corollary}

\begin{proof}
Since a user who waterfills between $m$ nodes transmits with zero power towards the rest of the nodes, his power allocation $p_{k}$ will belong to the interior of an $(m-1)$-dimensional face of $\strat_{k}$. The result then follows by combining this observation with Corollary \ref{cor:maxwaterfill}.
\end{proof}

We thus see that the Nash set $\eq$ has to be contained in the interior of a face of $\strat$ of dimension at most $A-1$. We will now show that $\eq$ is actually a singleton:

\begin{theorem}
The game $\game$ has a unique Nash equilibrium (a.s.).
\end{theorem}

\begin{proof}
The basic idea of the proof is, essentially, a geometric one. Indeed, assume that there are two distinct equilibrial points, $p$ and $p'$, sitting at the interior of some $d$-dimensional face $\rho$ of $\strat$. By convexity, the linear segment spanned by $p$ and $p'$ will also belong to the Nash set $\eq$ which implies that this line segment must (a.s.) lie in the subspace $W$ of degenerate directions of the game.

In other words, we see that if there is not a unique Nash equilibrium, then the subspace $W$ of degenerate directions intersects nontrivially with a $d$-dimensional face $\rho$ whose interior contains an equilibrium. However, since $\dim(W) = KA-A$ (a.s.), Corollary \ref{cor:face} gives $\dim(\rho)+\dim(W) \leq KA - A + A - 1<KA$. On the other hand, it is well-known that two generic subspaces of a real vector space $V$ intersect nontrivially if and only if the sum of their dimensions exceeds $\dim(V)$, so, since $\strat$ is  embedded in $\R^{KA}$, we may conclude that $\rho$ and $W$ intersect trivially with probability $1$, a contradiction.
\end{proof}


\section{Convergence to Equilibrium}

Having determined the properties of the game's unique equilibrium point, our task in this section will be to present a decentralized learning scheme which allows users to converge to this equilibrium point (and to estimate the speed of this convergence).


Given that the structure of the game $\game$ does not adhere to the (multilinear) setting of Nash \cite{Na51},\footnote{Or even the continuous population models of \cite{Sa01}: there are no ``node-specific'' rewards in our problem like the ``phenotype-specific'' growth rates of evolutionary biology.} the usual theory of evolutionary ``random-matching'' games does not apply either. This leaves on a rather unclear position on how to proceed, but since players invariably want to increase their rewards and an increase in payoff is equivalent to a decrease in potential, we will begin by considering the directional derivatives of the potential function $\Phi$:
\begin{equation}
v_{k\alpha}(p) \equiv -\frac{\pd \Phi}{\pd p_{k\alpha}} = \frac{b_{\alpha}g_{k\alpha}}{\sigma_{\alpha}^{2} + \insum_{\ell} g_{\ell\alpha} p_{\ell\alpha}}.
\end{equation}
Clearly, if a player transmits with positive power to node $\alpha$, then he will be able to calculate the gradient $v_{k\alpha}(p)$ in terms of the observables $p_{k\alpha}$ (the user's power allocation), $g_{k\alpha}$ (his channel gain coefficients), and \attn{the spectral efficiency $u_{k\alpha}(p) = b_{\alpha}\log\left(1+g_{k\alpha}p_{k\alpha}(\sigma_{\alpha}^{2} + \sum_{\ell\neq k} g_{\ell\alpha}p_{\ell\alpha})\right)$ of (\ref{eq:payoff}) which user $k$ observes at node $\alpha$}.\footnote{\attn{Note that this is different from gradient techniques applied to the utility functions themselves, a practice which requires the utility functions to be known.}} As a result, any learning scheme which relies only on the $v_{k\alpha}$'s will be inherently distributed in the sense that it only requires information that is readily obtainable by the individual players.

With all this in mind, a particularly simple scheme to follow is that of the replicator dynamics \cite{We95} associated with the ``marginal payoffs'' $v_{k\alpha}$. More specifically, this means that the players update their power allocations according to the differential equation:
\begin{equation}
\label{eq:RD}
\frac{dp_{k\alpha}}{dt}
\txs= p_{k\alpha} \big(v_{k\alpha}(p(t)) - v_{k}(p(t))\big),
\end{equation}
where $v_{k}$ is just the user average $v_{k}(p) = P_{k}^{-1}\insum_{\beta} p_{k\beta} v_{k\beta}(p)$.

As usual, the rest points of (\ref{eq:RD}) are characterized by the (waterfilling) property that, for every pair of nodes $\alpha,\beta\in\supp(p)$ to which user $k$ allocates positive power, we will also have $v_{k\alpha}(p) = v_{k\beta}(p)$. Hence, comparing this to the KKT conditions (\ref{eq:KKT}), we immediately see that the Nash equilibria of $\game$ are stationary in the replicator equation (\ref{eq:RD}). This result is well-known in finite Nash games with multilinear payoffs \cite{FL98} and in continuous population games \cite{Sa01}, but the converse does not hold: for instance, every vertex of $\strat$ is stationary in (\ref{eq:RD}), so stationarity of (\ref{eq:RD}) does not imply equilibrium.

\smallskip

Nevertheless, {\em only} Nash equilibria can be attracting, and, in fact, they attract almost every replicator solution orbit:

\begin{theorem}
\label{thm:convergence}
Let $q\in\strat$ be the unique (a.s.) equilibrium of $\game$. Then, every solution orbit of the replicator dynamics (\ref{eq:RD}) which begins at finite Kullback-Leibler entropy from $q$ will converge to it.

Furthermore, even if the game does not admit a unique equilibrium, every interior trajectory still converges to a Nash equilibrium (and not merely to the Nash set of the game).
\end{theorem}

\begin{remark*}
Recall that the {\em Kullback-Leibler divergence} (or {\em relative entropy}) of $p$ with respect to $q$ is \cite{We95}:
\begin{equation}
H_{q}(p)
= \insum_{k} H_{q_{k}}(p_{k})
= \insum_{k,\alpha} q_{k\alpha} \log\left({q_{k\alpha}}\big/{p_{k\alpha}}\right).
\end{equation}
Clearly, $H_{q}(p)$ is finite if and only if $p_{k}$ allocates positive power $p_{k\alpha}>0$ to all nodes $\alpha\in\supp(q)$ which are present in $q_{k}$; more succinctly, the domain of $H_{q}$ consists of all power allocations which are absolutely continuous w.r.t. $q$.
\end{remark*}

\begin{figure*}
\label{fig:convergence}
\noindent\negspace
\subfigure[Global convergence to equilibrium]{%
\label{subfig:convirr}
\includegraphics[width=0.48\columnwidth]{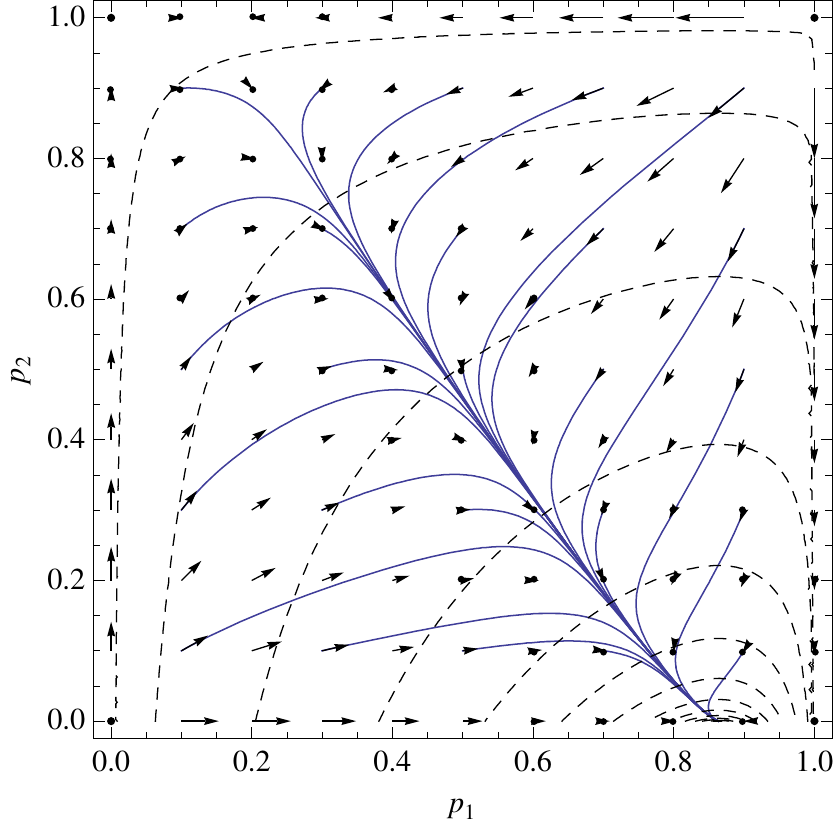}}
\hfill
\subfigure[Convergence in degenerate games]{
\label{subfig:convred}%
\includegraphics[width=0.48\columnwidth]{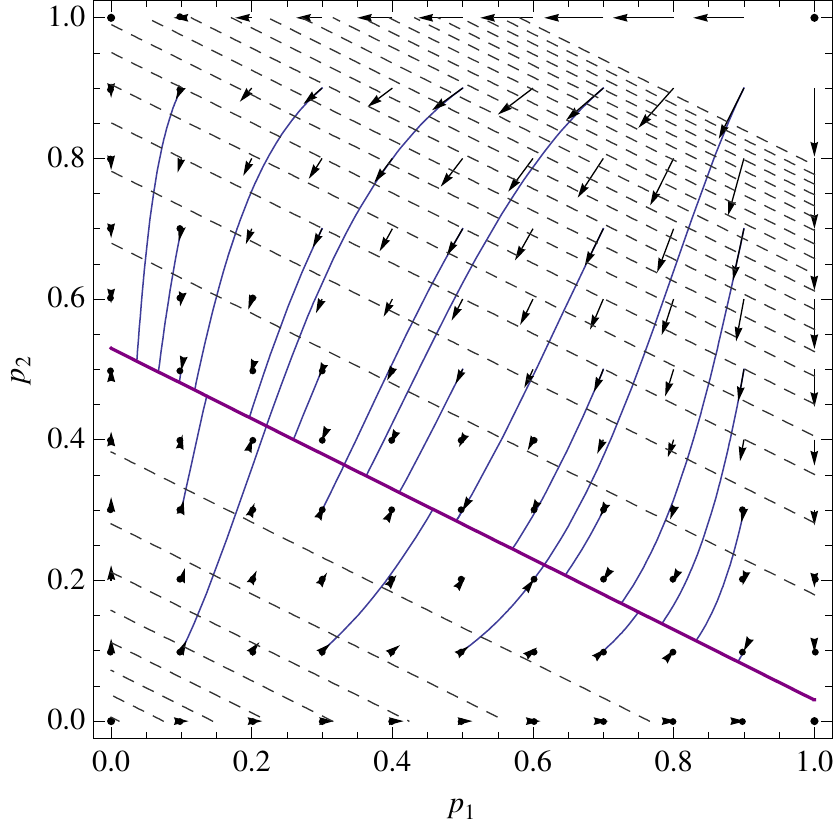}}
\caption{Convergence to equilibrium in 2x2 power allocation games (the dashed grey contours represent the level sets of the game's potential). If the game admits a unique equilibrium (which is almost always the case), then this equilibrium is (almost) globally attracting (Fig.~\ref{subfig:convirr}). However, even when the game has more than one equilibria (Fig.~\ref{subfig:convred}), every interior replicator trajectory converges to an equilibrium point.}
\end{figure*}

This convergence result (proved in Appendix \ref{apx:convergence}) is extremely powerful because it shows that the network's users will eventually settle down to a stable state which discourages unilateral deviations, even though they only have local information at their disposal. The only case that is left open in the above theorem is what happens if the initial K-L entropy of the solution orbit is infinite, i.e. if the users' initial power allocation does not support all of the nodes which are present in equilibrium. If this is the case, then the face-invariance property of the replicator dynamics ($p_{k\alpha}(t)=0$ iff $p_{k\alpha}=0$) will prevent the users from settling down to a Nash equilibrium. However, an easy analysis shows that if one takes the reduced game where each user only has access to the nodes to which he initially allocates positive power, then the users will actually converge to an equilibrium of this reduced game:
\attn{
\begin{proposition}
Let $p(0) = (p_{1}(0),\dotsc,p_{K}(0))$ be an initial power allocation profile in the game $\game$ and let $\act_{k}=\supp(p_{k}(0))\subseteq\act$. Then, if $\game^{0}$ is a reduced version of $\game$ which is played over $\strat\equiv\prod_{k} P_{k}\strat(\act_{k})$ with payoffs
\begin{equation}
u_{k}^{0}(p) = \sum_{k\in\act_{k}} b_{\alpha} \log\left(1 + \frac{g_{k\alpha}p_{k\alpha}}{\sigma_{\alpha}^{2} + \sum_{\ell\neq k} g_{\ell\alpha} p_{\ell\alpha}}\right),
\end{equation}
the replicator dynamics (\ref{eq:RD}) converge to the unique (a.s.) equilibrium of the reduced game $\game^{0}$.
\end{proposition}
}

\section{Conclusions}

In this paper, we studied distributed power allocation in parallel multiple access channels, modeling e.g. the problem of uplink communication in networks that consist of wireless receivers that operate in orthogonal frequency bands. Despite the fact that these games are special cases of the more general framework of \cite{SPB08i-sp,SPB08ii-sp,SPB08-it,SPB08-jsac}, the sufficient conditions provided therein for uniqueness of equilibrium typically fail in our case. Nonetheless, we show that the game {\em does} admit a unique equilibrium by studying the properties of the game's exact potential function (and correcting a mistake of \cite{PBLD09} in the process). Furthermore, by introducing a distributed learning scheme based on the replicator dynamics of evolutionary game theory, we show that users converge to the game's unique (a.s.) equilibrium. This result carries significant applicational potential because it ensures convergence to equilibrium even in decentralized settings where the users only have access to local information (in our case, the channel coefficients $g_{k\alpha}$ and the spectral efficiencies $u_{k\alpha}$).

Since the parallel MAC game is a special case of the more general IC one, a natural question that arises is whether our analysis extends to this more general case as well. One immediate observation is that the convergence properties of the replicator dynamics are still valid in general convex potential games played over products of simplices, but since the MIMO MAC game is actually played over the polytope of non-negative definite channel matrices with constrained trace, it is not as easy to write a continuous-time equation there. Further directions to be explored include the speed of convergence of the replicator dynamics to equilibrium (which can be shown to be exponentially fast) and the robustness of the replicator dynamics under stochastic disturbances which reflect inaccuracies in the users' observations (e.g. of the channel coefficients $g_{k\alpha}$). However, a disicussion of these issues would take us too far afield (and well beyond the space limitations of this paper), so we prefer to postpone them for the future.




\appendix


\section{Proof of Convergence}
\label{apx:convergence}

This appendix is devoted to the proof of Theorem \ref{thm:convergence}. The basic idea will be to show that the replicator dynamics (\ref{eq:RD}) admit a Lyapunov function, i.e. a non-negative function $f:\strat\to\R$ with $f(p)=0$ if and only if $p=q$ and such that $\dot f (p) \equiv \insum_{k,\alpha} \pd f/\pd p_{k\alpha} \,\dot p_{k\alpha}<0$ for all $p\neq q$.

A particularly appealing candidate is the game's own potential function $\Phi$. Indeed, an easy differentiation of (\ref{eq:potential}) yields $\pd\Phi/\pd p_{k\alpha} = - \pd u_{k}/\pd p_{k\alpha}$, so we obtain:
\begin{flalign}
\frac{d\Phi}{dt}
&= -\insum_{k,\alpha} \frac{\pd u_{k\alpha}}{dp_{k\alpha}} \frac{dp_{k\alpha}}{dt}\notag\\
&= \insum_{k,\alpha} v_{k\alpha}(p(t)) p_{k\alpha}(t) \left(v_{k\alpha}(p(t)) - v_{k}(p(t))\right)\notag\\
&= \insum_{k} P_{k}\left[\insum_{\alpha}\frac{p_{k\alpha}(t)}{P_{k}} v_{k\alpha}^{2}(p(t)) - v_{k}^{2}(p(t))\right]\leq 0,
\end{flalign}
by Jensen's inequality (recall that $\insum_{\alpha} p_{k\alpha} = P_{k}$). Since this inequality is strict if all the $p_{k\alpha}$ are positive and $v_{k\alpha}\neq v_{k\beta}$ for $\alpha\neq\beta$, this proves convergence to equilibrium when the game only has a unique equilibrium and the game has no degeneracy.

To get the more general case (and, also, for independent interest), it is much more instructive to consider as a Lyapunov candidate the relative entropy $H_{q}$ itself. Indeed, a simple differentiation gives:
\begin{equation}
\frac{d H_{q}}{dt}
= -\insum_{k,\alpha} \frac{q_{k\alpha}}{p_{k\alpha}(t)} \frac{d p_{k\alpha}}{dt}
=-\insum_{k,\alpha} q_{k\alpha} \big(v_{k\alpha}(p(t)) - v_{k}(p(t))\big),
\end{equation}
and, after rearranging the last term, we get:
\begin{equation}
\frac{dH_{q}}{dt} = \insum_{k,\alpha}(p_{k\alpha}(t) - q_{k\alpha})\,v_{k\alpha}(p(t)) \equiv - L_{q}(p(t)),
\end{equation}
where
\begin{equation}
\label{eq:evindex}
L_{q}(p) = -\insum_{k,\alpha} (p_{k\alpha} - q_{k\alpha}) \,v_{k\alpha}(p).
\end{equation}

%
We are thus left to show that $L_{q}(p(t))>0$ and, to that end, the key observation is that $L_{q}$ may be interpreted as a directional derivative of $\Phi$. So, let us set $f(\theta) = \Phi(q+\theta z)$, where $\theta\geq 0$ and $z$ is a vector in the (solid) {\em tangent cone} $\cone_{q}\strat$ of $\strat$ at $q$:
\begin{equation}
\label{eq:hfrelation}
\cone_{q}\strat \equiv \{z\in Z: z_{k\alpha}\geq0 \text{ for all } \alpha \text{ with } q_{k\alpha}=0\},
\end{equation}
i.e. $\cone_{q}\strat$ consists of those tangent directions $z\in Z$ which point towards the interior of $\strat$ (recall that $q$ might lie on the boundary of $\strat$). Clearly then, (\ref{eq:evindex}) may be rewritten as:
\begin{equation}
\label{eq:fprime}
f'(\theta) = \insum_{k,\alpha} \left.\frac{\pd\Phi}{\pd p_{k\alpha}}\right|_{q+\theta z} z_{k\alpha} = \theta^{-1} L_{q}(q+\theta z), 
\end{equation}
for all sufficiently small $\theta>0$ such that $q+\theta z\in\strat$.

However, since $q$ is the unique minimum of $\Phi$ (a.s.), $f(\theta)$ will be convex along any direction $z\in\cone_{q}\strat$, so that $\theta f'(\theta)\geq f(\theta) - f(0)$. Hence, if $p=q+\theta z$ is an arbitrary point of $\strat$, equations (\ref{eq:hfrelation}) and (\ref{eq:fprime}) yield the growth estimate:
\begin{equation}
\label{eq:evindexestimate}
L_{q}(p) = \theta f'(\theta) \geq f(\theta) - f(0) = \Phi(p) - \Phi(q).
\end{equation}
This last estimate shows that $L_{q}(p)\geq0$ for all $p\neq q$, thus concluding our proof of Theorem \ref{thm:convergence} for the non-degenerate case (note that then $\Phi(p) - \Phi(q)>0$ for all $q\neq p$).

To tackle the degenerate case, a semi-definite Lyapunov function (such as the game's potential $\Phi$ or the relative entropy $H_{q}$) is not enough because it ensures convergence to the set of minimum points and not to an actual point.
Clearly, the replicator dynamics in degenerate games might, in principle, exhibit phenomena of this kind. However, there is much more at work in (\ref{eq:RD}) than a single semi-definite Lyapunov function: there exists a whole \emph{family} of such functions, one for each equilibrium $q$.\footnote{This is also the reason that the relative entropy is a more suitable Lyapunov candidate: the potential has the same value at the {\em entire} Nash set of the game, while the relative entropy with respect to a point only vanishes at the point itself.}

To take advantage of this, it will be useful to shift our point of view to the {\em evolution function} $\Theta(p,t)$ of the dynamics (\ref{eq:RD}) which describes the solution trajectory that starts at $p$ at time $t=0$ and which satisfies the consistency condition:
\begin{equation}
\Theta(p,t+s) = \Theta(\Theta(p,t),s) \text{ for all $t,s\geq0$ and $p\in\strat$.}
\end{equation}
So, fix some initial condition $p\in\Int(\strat)$ (or, more generally, $p\in\strat_{q}$ where $\strat_{q}$ is the domain of the relative entropy function $H_{q}$) and let $p(t)=\Theta(x,t)$ be the corresponding solution orbit. If $q\in\eq$ is Nash, then, in view of the above discussion, the function $V_{q}(t) \equiv H_{q}(\Theta(p,t))$ will be decreasing (though, perhaps, not strictly so) and will converge to some $m\geq0$ as $t\to\infty$. It  thus follows that $p(t)$ converges itself to the level set $H^{-1}_{q}(m)$.

Suppose now that there exists some increasing sequence of times $t_{n}\to\infty$ such that $p_{n}\equiv p(t_{n})$ does not converge to the Nash set $\eq\equiv\eq(\game)$. By compactness of $\strat$ (and by descending to a subsequence if necessary), we may assume that $p_{n}=\Theta(p,t_{n})$ converges to some $p^{*}\notin \eq$ (but necessarily in $H^{-1}_{q}(m)$). Hence, for any $t>0$:
\begin{equation}
H_{q}(\Theta(p,t_{n}+t))
= H_{q}(\Theta(\Theta(p,t_{n}),t))
\to H_{q}(\Theta(p^{*},t))
< H_{q}(p^{*}) = m
\end{equation}
where the (strict) inequality stems from the fact that $\dot H_{q}<0$ outside $\eq$. On the other hand, $H_{q}(\theta(p,t_{n}+t))=V_{q}(t_{n}+t)\to m$, a contradiction.

Since the sequence $t_{n}$ was arbitrary, this shows that $p(t)$ converges to the set $\eq$. So, let $q'$ be a limit point of $p(t)$ with $p(t'_{n})\to q'$ for some sequence of times $t_{n}'\to\infty$. Then, $V_{q'}(t_{n}')=H_{q'}(p(t_{n}'))$ will converge to zero and, with $V_{q'}$ decreasing, we will have $\lim_{t\to\infty} V_{q'}(t) =0$ as well. Seeing as $H_{q'}$ only vanishes at $q'$, we conclude that $p(t)\to q'$, i.e. every interior trajectory converges to equilibrium.

\begin{figure}
\centering
\begin{tikzpicture}[scale=1.25]
\coordinate (A) at (-1.5,0);
\coordinate (B) at (0,2.598);
\coordinate (C) at (1.5,0);
\coordinate (P1) at ($ 4/5*(A) + 1/5*(B)$);
\coordinate (P2) at ($ 2/5*(B) + 3/5*(C)$);
\coordinate (Q) at ($1/3*(P1) + 2/3*(P2)$);
\draw (A)--(B);
\draw (B)--(C);
\draw (C)--(A);
\draw (P1)--(P2);
\draw[dashed, rotate=30] (Q) ellipse (10pt and 15pt);
\node [fill=black, shape =circle, inner sep=1pt, label=-45: $q$] at (Q) {};
\node at ($ 0.9*(B) + 0.25*(C)$) {$\strat$};
\node [inner sep=-2pt, label=-90: $\eq$] at ($ 0.8*(P1) + 0.2*(P2)$) {};
\node [inner sep=0pt] at (0,1.55) {$H^{-1}_{q}(m)$};
\end{tikzpicture}
\caption{The sets in the proof of Theorem \ref{thm:convergence}.}
\label{fig:convproof}
\end{figure}
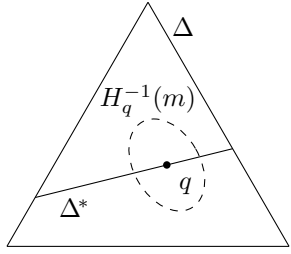

\bibliographystyle{ieeetran}
\bibliography{IEEEabrv,Bibliography,biblio_27-oct-2010}

\end{document}